\documentclass[11pt]{article}
\usepackage[centertags]{amsmath}
\usepackage{amsfonts}
\usepackage{amssymb}
\usepackage{amsthm}
\usepackage{amsmath}
\usepackage{tabularx}
\usepackage{bbm}
\usepackage{version}
\usepackage{placeins}
\usepackage{subfigure}
\usepackage{slashbox}

\usepackage[german,english]{babel}
\usepackage[cp850]{inputenc}
\usepackage{latexsym,graphicx}
\usepackage{color}

\theoremstyle{plain}
\newtheorem{thm}{Theorem}[section]

\newtheorem{lem}[thm]{Lemma}
\newtheorem{prop}[thm]{Proposition}
\theoremstyle{definition}
\newtheorem{defn}{Definition}[section]
\theoremstyle{remark}
\newtheorem{rem}{Remark}[section]

\DeclareMathOperator*{\esssup}{ess\,sup}

\author{
{\sc Olivier Bardou}\thanks{Corresponding author. Gaz de France,
Research and Development Division, 361 Avenue du Pr\'esident Wilson
- B.P. 33, 93211 Saint-Denis La Plaine cedex. E-mail: {\tt
olivier-aj.bardou@gazdefrance.com}, Phone: {\tt+33 1 49 22 54 58},
Fax: {\tt +33 1 49 22 57 10}} \quad{\sc Sandrine
Bouthemy}\thanks{Gaz de France, Research and Development Division,
361 Avenue du Pr\'esident Wilson - B.P. 33, 93211 Saint-Denis La
Plaine cedex. E-mail: {\tt sandrine.bouthemy@gazdefrance.com}} \quad
{\sc and} \quad {\sc Gilles Pag\`es}
\thanks{Laboratoire de Probabilit\'es et Mod\`eles al\'eatoires,
UMR~7599, Universit\'e Paris 6, case 188, 4, pl. Jussieu, F-75252
Paris Cedex 5, France. E-mail:{\tt  gpa@ccr.jussieu.fr}}
}
\date{6th April 2007}
\title{Optimal quantization for the pricing of swing options}

\begin{document}
\maketitle

\begin{abstract}
In this paper, we investigate a numerical algorithm for the pricing
of swing options, relying on the so-called optimal quantization
method.

The numerical procedure  is described in details and numerous
simulations are provided to assert its efficiency. In particular, we
carry out a comparison with the Longstaff-Schwartz algorithm.
\end{abstract}

\noindent {\em Key words: Swing options, stochastic control, optimal
quantization, energy.}

%
\noindent
\section*{Introduction}
In increasingly deregulated energy markets, swing options arise as
powerful tools for modeling supply contracts~\cite{geman-05}. In
such an agreement between a buyer and a seller, the buyer always has
to pay some amount even if the service or product is not delivered.
Therefore, the buyer has to manage his contract by constantly
swinging for one state to the other, requiring delivery or not. This
is the kind of agreement that usually links an energy producer to a
trader. Numerous other examples of energy contracts can be modeled
as swing options. From storages \cite{swing,carmona-ludkovski-06} to
electricity supply \cite{keppo-02,carmona-ludkovski-07}, this kind
of financial device is now widely used. And it has to be noticed
that its field of application has recently been extended to the IT
domain \cite{clearwater-huberman-06}.

 Nevertheless, the pricing of swings
remains a real challenge. Closely related to a multiple stopping
problem \cite{carmona-touzi-04,carmona-dayanik-06}, swing options
require the use of high level numerical schemes. Moreover, the high
dimensionality of the underlying price processes and the various
constraints to be integrated in the model of contracts based on
physical assets such as storages or gas fired power plants increase
the difficulty of the problem.

Thus, the most recent technics of mathematical finance have been
applied in this context; from trees to Least Squares Monte Carlo
based methodology \cite{THO,jaillet-ronn-al-04,lari-simchi-01},
finite elements \cite{winter-wilhelm-06} and duality approximation
\cite{MEHA}. But none of these algorithms gives a totally satisfying
solution to the valuation and sensitivity analysis of swing
contracts.

The aim of this paper is then to introduce and study a recent
pricing method that seems very well suited to the question. Optimal
Vector Quantization has yet been  successfully applied to the
valuation of multi-asset American Options \cite{spa,bernouilli,mf}.
It turns out  that this numerical technique is also very efficient
in taking into account the physical constraints of swing contracts.
For sake of simplicity we shall focus on gas supply contracts. After
a brief presentation of such agreements and some background on
Optimal Quantization methods \cite{handbook}, we show that a careful
examination of the properties of the underlying price process can
dramatically improve the efficiency of the procedure, as illustrated
by several numerical examples.

The paper is organized as follows: in the first section, we describe
in details the technical features of the supply contracts (with firm
or penalized constraints) with an emphasis on the features of
interest in view of a numerical  implementation: canonical
decomposition and normal form, backward dynamic programming of the
resulting stochastic control problem, existence of bang-bang
strategies for some appropriate sets of   local and global purchased
volume constraints. Section~2 is devoted to some background on
optimal vector quantization. In Section~3, our algorithm is briefly
analyzed and the {\em a priori} error bound established in the
companion paper~\cite{swingquantif} is stated (as well as the
resulting convergence result of the quantized premium toward the
true one). In Section~4, numerous simulations are carried out and
the quantization method is extensively compared to the well-known
least squares regression algorithm ``\`a la Longstaff-Schwartz". An
annex explains in details how the price processes we consider in
this paper can be quantized in the most efficient way.

\section{Introduction to swing options}
\subsection{Description of the contract}
A typical example of swing option is an energy (usually gas or
electricity) supply contract with optional clauses on price and
volume. The owner of such a contract is allowed to purchase some
amount of energy $q_ {t_k}$ at time $t_k,\,k=0,\dots,n-1$ until the
contract maturity $t_n=T$, usually one year. The purchase price
$K_{k}$ called strike price may be constant or indexed to past
values of crude oil. Throughout the paper we will consider that the
strike prices are constant and equal to $K$ over the term of the
contract. The volume of gas $q_{t_k}$ purchased at time $t_i$ is
subject to the local constraint
$$
q_{min} \leq q_{t_k} \leq q_{max}.
$$
The cumulative volume purchased prior to time $t_k$($i.e.$ up to $t_{k-1}$)  is defined by
$Q_{t_k} = \sum_{\ell=0}^{k-1} q_{t_\ell}$. It must satisfy the following
global constraint (at maturity):
$$
Q_T =\sum_{k=0}^{n-1} q_{t_k}\in [Q_{min},Q_{max}].
$$

Two approaches can be considered:

\smallskip
--  The constraints on the global purchased volumes are firm.

\smallskip
-- A penalty is applied if the constraints are not satisfied.

\medskip
The price at time $t$ of the forward contract delivered at time $T$
is denoted by $F_{t,T}$, $(F_{0,t_k})_{0 \leq k \leq n}$ being a
deterministic process (the future prices at time $0$) available and tradable on the market.

Let $(S_{t_k})_{0 \leq k \leq n}$ be the underlying Markov price
process defined on a probability space
$(\Omega,\mathcal{A},\mathbb{P})$. Note that it  can
be the observation at time $t_k,\,k=0,\dots,n$ of a continuous time
process. Ideally $S_t$ should be the spot price process of the gas
$i.e.$ $S_t = F_{t,t}$. However it does not correspond to a tradable
instrument which leads to consider in practice the day-ahead
contract $F_{t,t+1}$.


We consider its (augmented) natural filtration
$\mathcal{F}^S = (\mathcal{F}_{t_k}^S)_{0 \leq k \leq
n}$. The decision sequence $(q_{t_k})_{0 \leq k \leq n-1}$ is
defined on $(\Omega,\mathcal{A},\mathbb{P})$ as well and is
$\mathcal{F}^S$-adapted, $i.e$ $q_{t_k}$ is
$\mathcal{F}_{t_k}=\sigma(S_{t_0},\dots,S_{t_k})$ measurable,
$k=0,\dots,n$. At time
$t_k$ the owner of the contract gets $q_{t_k}(S_{t_k}-K)$.\\

\begin{rem}
The results of this paper can also be applied to every physical
asset or contract where the owner reward for a decision $q_{t_k}$ is
a function $\psi(t_k,q_{t_k},S_{t_k})$. In the case of supply
contracts, $\psi(t_k,q_{t_k},S_{t_k})=q_{t_k}(S_{t_k}-K)$. As for a
storage, $q_{t_k}$ represents the amount of gas the owner of the
contract decides to inject or withdraw and the profit at each date
is then
\begin{equation*}
\psi(t_k,q_{t_k},S_{t_k})=\left\{
  \begin{array}{llll}
    -q_{t_k}(S_{t_k} + c_I) &if & q_{t_k}\geq 0 & (Injection)\\
    -q_{t_k}(S_{t_k} - c_W) & if & q_{t_k}\leq 0 & (Withdrawal)\\
   0& if &q_{t_k}=0& (Same \  level \ in \  the \ storage)
  \end{array}
\right.
\end{equation*}
where $c_I$ (resp. $c_W$) denotes the injection (resp. withdrawal)
cost \cite{swing}.
\end{rem}

\subsubsection{Case with penalties}
We first consider that the penalties are applied at time $T$ if the
terminal constraint is violated. For a given consumption strategy
$(q_{t_k})_{0\leq k < n}$, the price is given by at time $0$
$$
P(0,S_0, 0)=\mathbb{E}\left(\sum_{k=0}^{n-1}e^{-rt_k}q_{t_k}(S_{t_k} - K) + e^{-rT}P_{_T}(S_T,Q_T)|{\cal F}_0\right)
$$
where $r$ is the interest rate. The function $(x,Q)\mapsto
P_{_T}(x,Q)$ is the penalization: $P_{_T} (x,Q)\leq 0$ and  $P_{_T}
(x,Q)|$ represents the sum that the buyer has to pay if global
purchased volume constraints, say $Q_{\min}$ and $Q_{\max}$, are
violated. \cite{swing} have already investigated this kind of
contract.

Then for every non negative $\mathcal{F}_{t_{k-1}}$ measurable
random variable $Q_{t_k}$ (representing  the cumulated purchased volume up to $t_{k-1}$), the price of the contract at time
$t_k,k=0,\dots,n-1$, is given by
\begin{equation}
P(t_k,S_{t_k},Q_{t_k}) = \esssup_{(q_{t_\ell})_{k\leq \ell <
n}}\mathbb{E}\left(\sum_{\ell=k}^{n-1}e^{-r(t_\ell-t_k)}q_{t_\ell}(S_{t_\ell} -
K) + e^{-r(T-t_k)}P_{_T}(S_T,Q_T)|S_{t_k}\right). \label{Eq:Prix}
\end{equation}

The standard penalization function is as follows:
\begin{equation}\label{penal}
P_{_T}(x,Q) = -\left(A\,x\,(Q - Q_{min})_- + B\,x\,(Q - Q_{max})_+\right)
\end{equation}
where $A$ and $B$ are large enough --~often equal~-- positive real constants.

\subsubsection{Case with firm constraints}
If we consider that constraints cannot be violated, then for every
non negative $\mathcal{F}_{t_{k-1}}$ measurable random variable
$Q_{t_k}$ defined on $(\Omega,\mathcal{A},\mathbb{P})$, the price of
the contract at time $t_k,k=0,\dots,n-1$ is given by:
\begin{equation}
P(t_k,S_{t_k},Q_{t_k})= \hskip -0.5 cm  \esssup_{(q_{t_\ell})_{k\leq \ell \leq n-1} \in
\mathcal{A}^{Q_{min},Q_{max}}_{k,Q_{t_k}}}\hskip -0.25 cm
\mathbb{E}\hskip -0.15 cm\left(\sum_{\ell=k}^{n-1}e^{-r(t_\ell-t_k)}q_{t_\ell}(S_{t_\ell}-K)|S_{t_k}\right).
\end{equation}
where
$$
\mathcal{A}^{Q_{min},Q_{max}}_{k,Q} = \left\{ (q_{t_\ell})_{k\leq
\ell \leq
n-1},\,q_{t_\ell}:(\Omega,\mathcal{F}_{t_\ell},\mathbb{P})\mapsto
[q_{min},q_{max}],\,  \sum_{\ell=k}^{n-1} q_{t_\ell} \!\in
[\left(Q_{min}-Q\right)_+, Q_{max}-Q]\right\}.
$$
At time $0$, we have:
$$
P(0,S_0,0)=\sup_{(q_{t_k})_{0\leq k \leq n-1} \in
\mathcal{A}^{Q_{min},Q_{max}}_{0,0}}\mathbb{E}\left(\sum_{k=0}^{n-1}e^{-rt_k}q_{t_k}(S_{t_k}-K)\right).
$$
Note that this corresponds to the limit case of the contract with
penalized constraints when $A=B=+\infty$. Furthermore, one shows
that when the penalties $A,B\to +\infty$  in~(\ref{penal}), the
``penalized" price converge to the ``firm" price. This has been
confirmed by extensive numerical implementations of both methods. In
practice when $A,\,B\approx 10\,000$ both methods become
indistinguishable for usual values of the volume constraints.

\subsection{Canonical decomposition  and  normalized contract}
\label{Section:DecompSwing} In  this section we obtain a
decomposition of the payoff of our swing contract (with firm
constraints) into two parts, one having a closed form expression. It
turns out that this simple decomposition leads to an impressive
increase of the precision of the price computation. It plays the
role of a variance reducer. Moreover, its straightforward financial
interpretation leads to a better understanding of the swing
contract.

In fact, we can distinguish a swap part and a normalized swing part:
\begin{eqnarray}\label{decompswing}
P(0,S_0)&=&
\underbrace{\mathbb{E}\left(\sum_{k=0}^{n-1}q_{min}e^{-rt_k}(S_{t_k}-K)\right)}_{Swap}\\
\nonumber &&+ \left(q_{max}-q_{min}\right) \underbrace{\sup_{(q_{t_k}) \in
\mathcal{A}_{[0,1]}^{\tilde{Q}_{min},\tilde{Q}_{max}}(0,0)}\mathbb{E}\left(\sum_{k=0}^{n-1}e^{-rt_k}q_{t_k}
(S_{t_k}-K)\right)}_{Normalized
\ Contract}
\end{eqnarray}
where
$$
\mathcal{A}_{[0,1]}^{\tilde{Q}_{min},\tilde{Q}_{max}}(k,Q) = \{
(q_{t_\ell})_{k\leq \ell \leq n-1},
q_{t_\ell}:(\Omega,\mathcal{F}_{t_\ell},\mathbb{P})\mapsto [0,1],
\sum_{\ell=k}^{n-1} q_{t_\ell} \in
[\left(\tilde{Q}_{min}-Q\right)_+,\tilde{Q}_{max}-Q]\}
$$
and
\begin{equation}
\tilde{Q}_{min} = \frac{(Q_{min}-nq_{min})_+}{q_{max}-q_{min}},\qquad
\tilde{Q}_{max}=\frac{(Q_{max}-nq_{min})_+}{q_{max}-q_{min}}.
\end{equation}

The price models investigated in the following sections define the
spot price as a process centered around the forward curve, and so
$\mathbb{E}(S_t) = F_{0,t}$ is known for every $t\!\in [0,T]$. Thus,
the swap part has a closed form given by
$$
Swap_0= q_{min}\sum_{i=0}^{n-1}e^{-rt_k}(F_{0,t_k}-K).
$$

The adaptation to contracts with penalized constraints is straightforward and amounts to modifying the penalization function in an appropriate way.

\subsection{Dynamic programming equation}
\label{Section:PgmDyn}
In \cite{swing}, it is shown that, in the penalized problem,
optimal consumption is the solution of a dynamic programming equation.
\begin{prop}Assume that for some positive constants p and C, the
following inequality holds for any $x>0$, and $Q \!\in
[n\,q_{min},n\,q_{max}]$:
$$
|P_{_T}(x,Q)| \leq C(1+x^p).
$$
 Then, there exists an optimal Markovian
consumption $q^*(t_k,S_{t_k},Q_{t_k})$ given by the maximum argument
in the following dynamic programming equation:
\begin{equation}\label{Eq:PgmDyn}
\left\{
  \begin{array}{l}
    P(t_k,S_{t_k},Q_{t_k})=\displaystyle \max_{q \in [q_{min},q_{max}]}\hskip -0.5 cm \left\{q(S_{t_k}-K)
    + e^{-r(t_{k+1}-t_k)}
\mathbb{E}(P(t_{k+1},S_{t_{k+1}},Q_{t_k}+q)|S_{t_k})\right\}, \\
    P(T,S_T,Q_T)= P_{_T}(S_T,Q_T).
  \end{array}
\right.
\end{equation}
\end{prop}
Usually, the function $P_{_T}(x,Q)$ is given by~(\ref{penal}). Then, the case with firm constraints corresponds to the limit case where  $P_{_T}(x,Q)=(-\infty)\mbox{\bf 1}_{\{x\notin [Q_{\min},Q_{max}]\}}$.

When considering a  contract with firm constraints, a more operating
form (see~\cite{swingquantif}) can be the following
  \begin{eqnarray}\label{Eq:PgmDynfirm}
  P(t_k,S_{t_k},Q_{t_k})\!\!&\!\!=\!\!&\!\!\displaystyle
\max \!\!\left\{\!q(S_{t_k}\!-\!K)
  +\mathbb{E}(P(t_{k+1},S_{t_{k+1}},Q_{t_k} \!+\! q)|S_{t_k}),\qquad \right.  \\
\nonumber && \left. q \!\in [q_{\min},q_{\max}],\;
Q_{t_k}+q\!\in[(Q_{\min}-(n-k)q_{\max})_+,(Q_{\max}-(n-k)q_{\min})_+]
\right\}.
\end{eqnarray}

\subsection{Bang Bang consumption}
\label{Section:BangBang}
\subsubsection{Case with penalties on purchased volumes}
 \cite{swing} showed
the  following theoretical result.
\begin{thm}
Consider the Problem~\ref{Eq:Prix} and $P_{_T}(x,Q)=-xP(Q)$, $P$ being a
continuously differentiable function. If the following condition
holds
$$
\mathbb{P}\left(e^{-rt_k}(S_{t_k}-K) +
\mathbb{E}(e^{-rT}S_TP'(Q^*_T)|S_{t_k},Q^*_{t_k})=0\right)=0,$$ the
optimal consumption at time $t_k$ is necessarily of bang-bang type
 given by
\begin{eqnarray*}
q^*(t_k,S_{t_k},Q_{t_k}^*) & = &
q_{max}\mbox{\bf 1}_{\{e^{-rt_k}(S_{t_k}-K) +
\mathbb{E}(e^{-rT}S_TP'(Q_T^*)|S_{t_k},Q_{t_k}^*)>0\}}\\
& & +q_{min}\mbox{\bf 1}_{\{e^{-rt_k}(S_{t_k}-K) +
\mathbb{E}(e^{-rT}S_TP'(Q_T^*)|S_{t_k},Q_{t_k}^*)<0\}}.
\end{eqnarray*}
\end{thm}
The above assumption seems difficult to check since it involves the unknown optimal consumption. However, this would be the case provided one shows that the random variable $e^{-rt_k}(S_{t_k}-K) +
\mathbb{E}(e^{-rT}S_TP'(Q^*_T)|S_{t_k},Q^*_{t_k})$ is absolutely continuous as noticed in~\cite{swing}.

\subsubsection{Case with firm constraints}
In the companion paper~\cite{swingquantif}, we establish some
properties of the value function of the  swing options viewed {\em
as a function of the global volume constraints} $(Q_{min},Q_{max})$.
Thanks to~(\ref{decompswing}) one may assume without loss of
generality that the contract is normalized, $i.e.$ $q_{min}=0$ and
$q_{max}=1$. We consider the following value function:
$$
P(Q_{min},Q_{max})=\sup_{(q_{t_k})_{0\leq k \leq n-1} \in
\mathcal{A}^{Q_{min},Q_{max}}_{0,0}}\mathbb{E}\left(\sum_{k=0}^{n-1}e^{-r
t_k}q_{t_k}(S_{t_k}-K)\right)
$$
defined on the unit (upper) simplex $\{(u,v)\!\in \mathbb{R}^2,\, 0\le u\le v\le n\}$.
\begin{prop}
The premium function $(Q_{min},Q_{max}) \mapsto P(Q_{min},Q_{max})$ is a concave, piecewise affine
function of the global purchased volume constraints, affine on elementary triangles
$(m,M) + \{(u, v), 0 \leq u \leq v \leq 1\}$, $(m, M) \in
\mathbb{N}^2$, $m \leq M \leq n$ and $(m,M) + \{(u, v), 0 \leq v \leq u \leq 1\}$, $(m, M) \in
\mathbb{N}^2$, $m \leq M-1 \leq n-1$ which   tile of the unit (upper) simplex.
\end{prop}

\begin{thm}
For integral valued global constraints, $i.e.$ $(Q_{min},Q_{max})\!\in \mathbb{N}^2$,
there  always exists a bang-bang optimal strategy $i.e.$ the {\em a priori}  $[0,
1]$-valued optimal purchased quantities $q_{t_k}^*$ are in fact
always equal to $0$ or $1$.
\end{thm}

\begin{rem} This result can be extended in some way to any couple of
global constraints when all the payoffs are nonnegative
(see~\cite{swingquantif}). Furthermore, it has nothing to do with
the Markov dynamics of the underlying  asset and holds in a quite
general abstract setting.
\end{rem}

An example of the premium function $(Q_{min},Q_{max}) \mapsto
P(Q_{min},Q_{max})$ is depicted on Figure~\ref{Fig:PriceSurface}.
\begin{center}
\begin{figure}[h!]
\begin{center}
  \includegraphics[width=12.5cm]{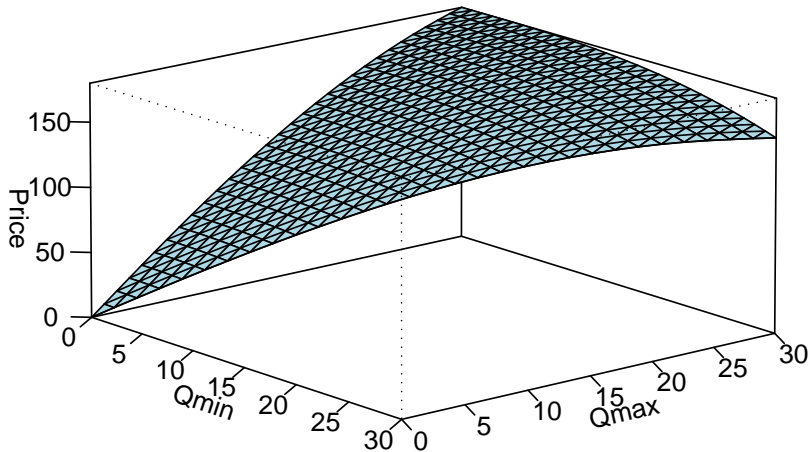}\\
  \caption{\em Value function $P(Q_{min},Q_{max})$ $versus$ the global constraints}\label{Fig:PriceSurface}
  \end{center}
\end{figure}
\end{center}
Now we  turn to the problem of the
numerical evaluation of such contracts. As announced, we focus on an
optimal quantization algorithm.

\section{Optimal quantization}
\label{Section:Quantization} Optimal Quantization
\cite{pages97,bernouilli,spa,mf} is a method coming from Signal
Processing devised to approximate a continuous signal by a discrete
one in an optimal way. Originally developed in the 1950's, it was
introduced as a quadrature formula for numerical integration in the
late 1990's, and for conditional expectation approximations in the
early 2000's, in order to price multi-asset American style options.

Let $X$ be an $\mathbb{R}^d$-valued random vector defined on a
probability space $(\Omega, \mathcal{F}, \mathbb{P})$. Quantization
consists in studying the best approximation of $X$ by random vectors
taking at most $N$ fixed values $x^1,\dots,x^N \in \mathbb{R}^d$.

\begin{defn} Let $x = (x^1,\dots,x^N) \in (\mathbb{R}^d)^N$. A partition
$(C_i(x))_{i=1,\dots,N}$ of $\mathbb{R}^d$ is a {\em Voronoi tessellation}
of the $N${-quantizer} $x$ (or {\em codebook}; the term  {\em grid} being used for $\{x^1,\ldots, x^N\}$) if, for every $i \in \{1,\dots,N\}$, $C_i(x)$ is
a Borel set satisfying
$$
C_i(x) \subset \{\xi \in \mathbb{R}^d, |\xi - x^i| \leq \min_{i \neq j}|\xi -
x^j|\}
$$
where $|\,.\,|$ denotes the canonical Euclidean norm on $\mathbb{R}^d$.
\end{defn}

The  nearest neighbour projection on $x$  induced by a Voronoi partition  is defined by
$$
{\rm Proj}_x: y \in \mathbb{R}^d \mapsto x^i\; \mbox{ if }  \;
y\!\in C_i(x).
$$
Then, we  define an {\em $x$-quantization} of $X$ by
$$
\hat{X}^x = {\rm Proj}_x(X).
$$
The {\em pointwise error} induced when  replacing $X$ by $\widehat
X^x$ is given by $|X-\widehat X^x|=
d(X,\{x^1,\ldots,x^N\})=\min_{1\le i\le N}|X-x^i|$. When $X$ has an
absolutely continuous distribution, any two $x$-quantizations are
$\mathbb{P}$-$a.s.$ equal.

\smallskip
The quadratic {\em mean quantization error} induced by the the $N$-tuple $x\!\in\mathbb{R}^d$ is defined
as the quadratic norm of the pointwise error $i.e.$ $\|X-\hat{X}^x\|_{_2}$.

We briefly recall some classical facts about theoretical and
numerical aspects of Optimal Quantization. For details we refer
$e.g.$ to \cite{graf,handbook}.
\begin{center}
\begin{figure}[h!]
\begin{center}
  \includegraphics[width=8cm]{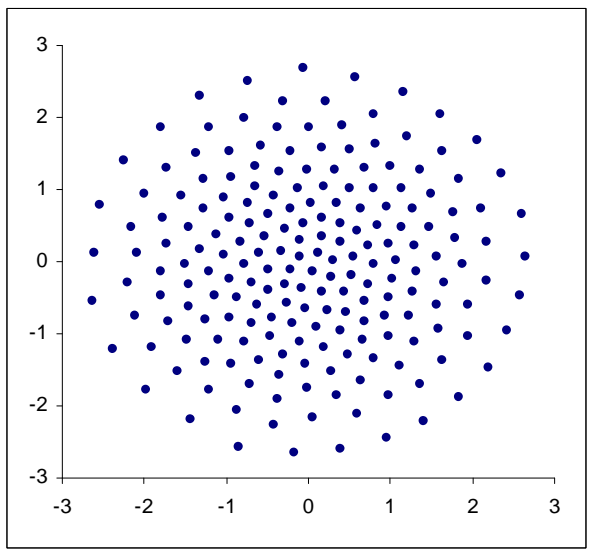}\\
  \caption{\em Optimal quadratic quantization of the normal distribution $\mathcal{N}(0,I_2)$,
$N=200$.}\label{Fig:Quantif2d}
  \end{center}
\end{figure}
\end{center}
\begin{thm}\cite{graf}
Let $X \!\in L^2(\mathbb{R}^d,\mathbb{P})$. The quadratic quantization error
function
$$
x=(x^1,\dots,x^N) \longmapsto
\mathbb{E}(\min_{1\leq i \leq N} |X - x^i|^2) =
\|X-\hat{X}^x\|_{_2}
$$
reaches a minimum at some quantizer $x^*$.
Furthermore, if the distribution $\mathbb{P}_X$ has an infinite
support then $x^{*,(N)}=(x^{*,1},\dots,x^{*,N})$ has pairwise distinct
components and $N \mapsto \min_{x \in
(\mathbb{R}^d)^N}\|X-\hat{X}^x\|_{_2}$ is   decreasing to $0$ as $N
\uparrow +\infty$.
\end{thm}

Figure~\ref{Fig:Quantif2d} shows a quadratic optimal quantization grid for a bivariate
normal distribution $\mathcal{N}(0,I_2)$. The convergence rate to $0$ of optimal quantization error is ruled by the
so-called Zador Theorem.
\begin{thm}
\cite{graf} Let $X \in L^{2+\delta}(\mathbb{P}),\delta > 0$, with
$\mathbb{P}_X(d\xi) = \varphi(\xi)\lambda_d(d\xi) + \nu(d\xi)$, $\nu
\perp \lambda_d$ ($\lambda_d$ Lebesgue measure on $\mathbb{R}^d$).
Then
$$
\lim_{N\rightarrow +\infty} (N^{\frac{2}{d}}\hskip -0.2 cm \min_{x\in (\mathbb{R}^d)^N} \|X-\hat{X}^x\|_{_2}) =
J_{2,d} \left(\int_{\mathbb{R}^d}\varphi^{\frac{d}{d+2}}d
\lambda_d\right)^{1+\frac{2}{d}}.
$$
\end{thm}

The true value of $J_{2,d}$ is unknown as soon as $d \geq 3$. One
only knows that $J_{2,d} = \frac{d}{2\pi e}+o(d)$.

Zador's Theorem implies that $\|X-\hat{X}^{x^{*,(N)}}\|_{_2} =
O(N^{-\frac{1}{d}})$ as $N \rightarrow + \infty$.

\begin{prop}\cite{pages97,handbook} Any
$L^2$-optimal quantizer $x\!\in \mathbb{R}^d$ satisfy the following stationarity property
$$
\mathbb{E}(X|\hat{X}^x) = \hat{X}^x.
$$
In particular, for any stationary quantizer $\mathbb{E}(X) =
\mathbb{E}(\hat{X}^x)$.
\end{prop}
The random vector $\hat{X}^x$ takes its value in a finite space $\{x^1,\dots,x^N\}$, so
for every continuous functional $f:\mathbb{R}^d\to \mathbb{R}$ with $f(X)\!\in L^2(\mathbb{P})$ , we have
$$
\mathbb{E}(f(\hat{X}^x)) =
\sum_{i=1}^N f(x^i) \mathbb{P}(X \in C_i(x))
$$
which is the quantization based quadrature formula to approximate
$\mathbb{E}\left(f(X)\right)$ \cite{pages97,handbook}. As
$\hat{X}^x$ is close to $X$ , it is natural to estimate
$\mathbb{E}(f(X))$ by $\mathbb{E}(f(\hat{X}^x))$ when $f$ is
continuous. Furthermore, when $f$ is smooth enough, on can upper
bound the resulting error using $\|X-\hat{X}^x\|_{_2}$, or even
$\|X-\hat{X}^x\|_{_2}^2$ (when the quantizer $x$ is stationary).

The same idea can be used to approximate the conditional expectation
$\mathbb{E}(f(X)|Y)$ by $\mathbb{E}(f(\hat{X})|\hat{Y})$, but one
also needs the transition probabilities:
$$
\mathbb{P}(X  \in C_j(x)|Y \in C_i(y)).
$$

The application of this technique to the quantization of spot price
processes is discussed in details in the Annex, page
\pageref{Section:FPQ}.

\section{Pricing swing contracts with optimal quantization}
\subsection{Description of the algorithm (general setting)}
In this section we assume that $(S_{t_k})_{0\le k\le n}$ is a Markov
process. For sake of simplicity, we consider that there is no
interest rate. We also consider a normalized contract, as defined in
Section \ref{Section:DecompSwing}.

In the penalized problem, the price of the swing option is given by the following dynamic programming equation (see Equation \ref{Eq:PgmDyn}): \\

$\left\{
  \begin{array}{l}
    P(t_k,S_{t_k},Q_{t_k})=\displaystyle \max_{q \in \{0,1\}}{[q(S_{t_k}-K) +
\mathbb{E}(P(t_{k+1},S_{t_{k+1}},Q_{t_k}+q)|S_{t_k})]} \\
\ \\
    P(T,S_T,Q_T)= P_{_T}(S_T,Q_T)
  \end{array}
\right.$\\
where $t_k = k \Delta,k=0,\dots,n$, $\Delta = \frac{T}{n}$.

The bang-bang feature of the optimal consumption (see Section
\ref{Section:BangBang}) allows us to limit the possible values of
$q$ in the dynamic programming equation to $q\!\in \{0,1\}$. At time
$t_k$, possible values of the cumulative consumption are
\begin{equation}Q_{t_k}^\ell = \ell , 0
\leq \ell  \leq k.\label{Eq:CumulativeVolume}
\end{equation}

At every time $t_k$ we consider a(n optimized) $N_k$-quantization
$\widehat S_{t_k}=\widehat S_{t_k}^{x_k^{(N)}}$, $k=0,\ldots,n$
based on an optimized quantization $N_k$-tuple (or grid)
$x_k^{(N_k)}:=\left(s_k^1,\dots, s_{k}^{N_k}\right)$ of the spot
$S_{t_k}$.

The modeling of the  future price by multi-factor Gaussian processes
with memory (see Section~4 for a toy example) implies that
$(S_t)_{t\in [0,T]}$ is itself a Gaussian process. Then the
quantization of $S_{t_k}$ can be obtained by a simple
dilatation-contraction (by a factor ${\rm StD}(S_{t_k})$) from
optimal quantization grids of the (possibly multivariate) normal
distribution, to be downloaded on the website~\cite{website06}
\[
\mbox{\tt www.quantize.maths-fi.com}
\]

Then we compute the price at each time $t_k$, for all points on the
corresponding grid, and for all the possible cumulative
consumptions:
\begin{equation}\label{Eq:QuantizedPgmDyn}\left\{
  \begin{array}{l}
    P(t_k,s_{k}^i,\hat{Q}_{t_k})=\displaystyle \max_{q \in \{0,1\}}{[q(s_{k}^i-K) +
\mathbb{E}(P(t_{k+1},\hat{S}_{t_{k+1}},\hat{Q}_{t_k}+q)|\hat{S}_{t_k}=s_{k}^i)]}\\
\hskip 2,1 cm i=1,\ldots,N_k, \\
    P(T,s_T^i,\hat{Q}_T)= P_{_T}(s_T^i,\hat{Q}_T),\; i=1,\ldots,N_n.
  \end{array}\right.
\end{equation}

When considering a contract with firm constraints, we need to
compute the price at each time $t_k$, for all the points of the
quantization grid of the spot price, and for all the admissible
cumulative consumptions (See Figure \ref{Fig:AdmissibleVolume})
\begin{equation}Q_{t_k}^\ell = \ell
+\left(Q_{\min}-(n-k+1)\right)_+,
\ell=0,...,\min(k,Q_{\max})-\left(Q_{\min}-(n-k+1)\right)_+.\label{Eq:CumulativeVolumeFirm}\end{equation}
\begin{center}
\begin{figure}[h!]
\begin{center}
  \includegraphics[width=8.5cm]{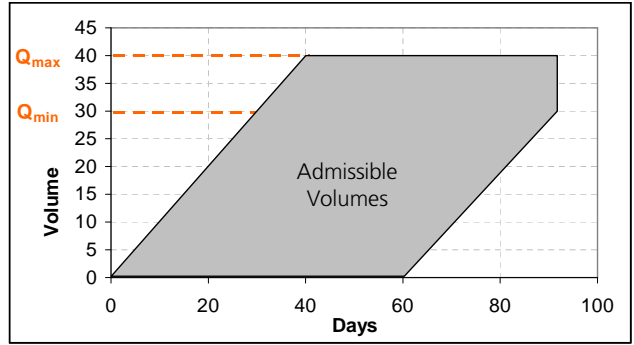}\\
   \caption{\em \label{Fig:AdmissibleVolume}Volume Constraints}
  \end{center}
\end{figure}
\end{center}
Using the bang-bang feature (See section \ref{Section:BangBang}) and
the dynamic programming principle (Equation~\ref{Eq:PgmDynfirm}),
this price is given by
\begin{eqnarray}\label{Eq:QuantizedPgmDynfirm}
 P(t_k,s_{k}^i,\hat{Q}_{t_k})\!\!&\!\!=\!\!&\!\!\displaystyle \max \!\!\left\{\!q(s_{k}^i\!-\!K)
    \!+\!\mathbb{E}(P(t_{k+1},\hat{S}_{t_{k+1}},\hat{Q}_{t_k} \!+\! q)|\hat{S}_{t_k}=s_{k}^i),\qquad \right.  \\
\nonumber && \left. \qquad q \!\in \{0,1\},\;
\hat{Q}_{t_k}+q\!\in[\left(Q_{\min}-(n-k)\right)_+,Q_{\max}]
\right\}.
\end{eqnarray}

Since $\hat{S}_{t_k}$ takes its values in a finite space, we can
rewrite the conditional expectation as:
$$
\mathbb{E}(P(t_{k+1},\hat{S}_{t_{k+1}},Q)|\hat{S}_{t_k}=s_{k}^i)
= \sum_{j=1}^{N_{k+1}}P(t_{k+1},s_{k+1}^j,Q) \pi^{ij}_k
$$
where
$$
\pi^{ij}_k = \mathbb{P}(\hat{S}_{t_{k+1}} = s_{k+1}^j | \hat{S}_{t_{k}} =
s_{k}^i)
$$
is the {\em quantized transition probability} between times $t_k$
and $t_{k+1}$. The whole set of quantization grids equipped with the
transition matrices make up the so-called ``quantization tree".  The
transition weights $(\pi_k^{ij})$ matrices are the second quantity
needed to process the quantized dynamic programming
principle~(\ref{Eq:QuantizedPgmDyn} or
\ref{Eq:QuantizedPgmDynfirm}). A specific fast parallel
 quantization procedure has been developed in our multi-factor Gaussian framework to speed up (and parallelize)
the computation of these weights (see Annex). In a more general
framework, one can follow the usual Monte Carlo approach described
in~\cite{bernouilli} to compute the grids and the transition of the
global quantization tree.

When $(S_{t_k})$ has no longer a Markov dynamics but appears as a
function of a Markov chain $(X_k)$: $S_{t_k}=f(X_k)$, one can
proceed as above, except that one has to quantize $(X_k)$. Then the
dimension of the problem (in term of quantization) is that of the
structure process $(X_k)$.

Of course one can always implement the above procedure formally
regardless of the intrinsic dynamics of $(S_{t_k})$. This yields to
a drastic dimension reduction (from that of $X$ downto that of
$(S_{t_k})$). Doing so, we cannot apply the convergence theorem (see
Section \ref{Section:Cvg}) which says that in a Markovian framework
the premium resulting from \eqref{Eq:QuantizedPgmDyn} or
\eqref{Eq:QuantizedPgmDynfirm} will converge toward the true one as
the size of the quantization grid goes to infinity.

This introduces a methodological residual error that can be compared
to that appearing in \cite{BAMA} algorithm for American option
pricing. However, one checks on simulations that this residual error
turns out to be often negligible (see~Section~\ref{dimred}).

\subsection{Complexity}

The first part of the algorithm consists in designing the
quantization tree and the corresponding weights. The complexity of
this step is directly connected to the size of the quantization
grids chosen for the transitions computation in 1-dimension, or to
the number of Monte Carlo simulations otherwise. However those
probabilities have to be calculated once for a given price model,
and then several contracts can be priced on the same quantization
tree. So we will mainly focus on the complexity of the pricing part.

We consider a penalized normalized contract, $i.e$ $q_{min}=0$ and
$q_{max}=1$. The implementation of the dynamic programming principle
requires three interlinked loops. For each time step $k$ (going
backward from $n$ to $0$), one needs to compute for all the points
$s_k^i$, $i=1,\ldots,N_k$ of the grid and for every possible
cumulative consumption $Q_{t_k}^\ell$ ($0\leq \ell \leq k$) (see
(\ref{Eq:CumulativeVolume})) the functional
$$
\max_{q \in \{0,1\}}{[q(s_k^i-K) +
\mathbb{E}(P(t_{k+1},\hat{S}_{t_{k+1}},Q_{t_k}+q)|\hat{S}_{t_k}=s_{k}^i)]
}
$$
which means computing twice a sum of $N_{k+1}$ terms.

Hence, the complexity is proportional to
$$
\sum_{k=0}^{n-1}(k+1)N_kN_{k+1}.
$$
In the case where all layers in the quantization tree have the same
size, $i.e$ $N=N_k,\forall k=1,\dots,n$, the complexity is
proportional to $\frac{n^2N^2}{2}$. This is not an optimal design
but {\em only one grid needs to be stored}. It is possible to reduce
the algorithm complexity by optimizing the grid sizes $N_i$
\footnote{To minimize the complexity, set $N_k \approx
\frac{2nN}{(k+1)log(n)},\,k=0,\dots,n-1$, which leads to a global
complexity proportional to $\frac{4n^2N^2}{\log(n)}$}(with the
constraint $\sum_kN_k = nN$), but it costs more memory space.\\

In the case of firm constraints, the dynamic programming principle
(\ref{Eq:QuantizedPgmDynfirm}) has to be computed for every
admissible cumulative consumption, $i.e$ for every $Q_{t_k}^\ell$
($0\leq \ell \leq
\min(k-1,Q_{\max})-\left(Q_{\min}-(n-k+1)\right)_+$, see
(\ref{Eq:CumulativeVolumeFirm})). The complexity is proportional to
$$
\sum_{k=0}^{n-1}(\min(k,Q_{\max})-\left(Q_{\min}-(n-k+1)\right)_++1)N_kN_{k+1}.
$$

The complexity in the case of firm constraints is lower than the one
for a penalized problem, and depends on the global constraints
$\left(Q_{\min},Q_{\max}\right)$. But the implementation is easier
in the case of a penalized problem, because one does not need to
check if the cumulative consumption volume is admissible. Both
approaches have been numerically tested and results are
indistinguishable for large enough penalties. For the implementation
readiness, the approach with penalties has been adopted.

In order to reduce the complexity of the algorithm, one usually
prunes the quantization tree. In most examples, at each layer $k$,
many terms of the transition matrix $(\pi_{ij}^k)_{i,j}$ are equal
to 0 or negligible. So while the transition probabilities are
estimated, all the transitions that are not visited are deleted.
This step is important because it allows to reduce significantly the
algorithm complexity.

In practice we can even neglect transitions whose probability is
very low, say less than $10^{-5}$.

\subsection{Convergence}
\label{Section:Cvg} In \cite{swingquantif} is proved an error bound
for the pricing of swing options by optimal quantization.

Let $P^n_0(Q)$ denote the price of the swing contract at time $0$.
$n$ is the number of time step, and $Q=(Q_{min},Q_{max})$ is the
global constraint. We consider a contract with normalized local
constraints, $i.e$ $q_{min}=0$ and $q_{max}=1$. The ``quantized"
price $\hat{P}^n_0(Q)$ is the approximation of the price obtained
using optimal quantization.

\begin{prop}
Assume there is a real exponent $p \in [1,+\infty)$ such that the
($d$-dimensional) Markov structure process $(X_k)_{0 \leq k \leq
n-1}$ satisfies
$$
\max_{0 \leq k \leq n-1} |X_k| \in
L^{p+\eta}(\mathbb{P}), \eta > 0.
$$
At each time $k \in \{0,\dots, n -
1\}$, we implement a (quadratic) optimal quantization grid $x^N$ of
size $N$ of $X_k$. Then
$$\parallel \sup_{Q \in T^+(n)} |P^n_0(Q) - \hat{P}^n_0(Q) |\parallel _p \leq
C\frac{n}{N^{\frac{1}{d}}}$$ where $T^+(n):= \{(u, v), 0 \leq u \leq
v \leq n\}$ is the set of admissible global constraints (at time
$0$).
\end{prop}

In fact this error bound turns out to be conservative and several
numerical experiments, as those presented in Section
\ref{Section:Numerical}, suggest that in fact the true rate (for a
fixed number $n$ of purchase instants) behaves like
$O(N^{-\frac{2}{d}})$.

\section{Numerical experiments}
\label{Section:Numerical} In this section the same grid size has
been used at each time step, $i.e.$ we always have $N_k = N,
k=0,\dots,n$. The results have been obtained by implementing the
penalized problem and using the canonical decomposition (see
Section~\ref{Section:DecompSwing}).

\subsection{The one factor model}
Swing options are often priced using the least squares regression
method ``\`a la Longstaff-Schwartz" \cite{ls}. This section aims to
compare our  numerical results to those obtained with
Longstaff-Schwartz method. We consider a one factor model, which
corresponds to a one dimensional Markov  structure process.

\subsubsection{Quantization tree for a one dimensional structure process}
\label{Section:1d} We consider the following diffusion model for the
forward contracts $(F_{t,T})_{0 \leq t \leq T}$:
$$\frac{dF_{t,T}}{F_{t,T}} = \sigma e^{-\alpha (T-t)}dW_t$$
where $W$ is a standard Brownian motion. It yields:
$$
S_t = F_{0,t}\exp{\left(\sigma
\int_0^t{e^{-\alpha(t-s)}dW_s} - \frac{1}{2}\Lambda_t^2\right)}
$$
where
$$
\Lambda_t^2 = \frac{\sigma^2}{2\alpha}(1-e^{-2\alpha t}).
$$
Denote $X_k = \int_0^{k \Delta}{e^{-\alpha(t-s)}}dW_s$. The
structure process $(X_k)_{k\geq 0}$ can be quantized using the fast
parallel quantization method described in the Annex (page
\pageref{Section:FPQ}). Let $x_k^{(N)}$ denote an (optimal)
quantization grid of $X_k$ of size $N$. We have to compute for every
$ k \!\in \{0,\dots,n-1\}$, and every $(i,j) \!\in \{1,\dots,N\}^2$,
the following (quantized transition) probabilities:
\begin{equation}p^{ij}_k=\mathbb{P}(\eta_1 \in
C_i(x_k^{(N)});\alpha_{k+1}\eta_1 + \beta_{k+1}\eta_2  \in
C_j(x_{k+1}^{(N)}))\label{Eq:Proba1d}\end{equation} where
$(\eta_1,\eta_2) \sim \mathcal{N}(0,I_2)$, and $\alpha_k$ and
$\beta_k$ are  scalar coefficients that can be explicited.

This can be done by using quantization again and importance sampling
as presented in the Annex, page \pageref{Section:AR1}
(see~(\ref{Eq:ImportanceSampling1D})).

\subsubsection{Comparison with the regression method}
We first use the following parameters for the one factor model:

\medskip
\centerline{$\sigma=70\%$, $\alpha=4$, $F_{0,t_k}=20,\;k=0,\dots,n$.}

\medskip
The following tables present the results obtained with Longstaff
Schwartz and optimal quantization, for different strike values.
$1000$ Monte-Carlo sample paths have been used for
Longstaff-Schwartz method and the confidence interval of the Monte
Carlo estimate is given in the table. A $100$-point grid has been
used to quantize the spot price process, and the transitions have
been computed with a $500$-point grid. The local volume constraints
$q_{min}$ and $q_{max}$ are set to $0$ and $6$  respectively .

We first consider a case without constraints (Table \ref{Tab:LS1}),
which means that the swing option is a strip of calls, whose price
is easily computed with the Black-Scholes formula.

\begin{center}
\begin{table}[h!]
\centering
\begin{tabular}{|l|c|c|c|c|}
    \hline
   & $K=5$ & $K=10$ & $K=15$ & $K=20$ \\
   \hline
  Longstaff-Schwartz &[32424,33154]  & [21360,22127]  & [11110,11824] & [3653,4109]\\
 \hline
  10 point grid & 32726 & 21806 & 11311 & 3905 \\
  20 point grid & 32751 & 21834 & 11367 & 3943 \\
  50 point grid & 32759 & 21843 & 11380 & 3964 \\
  100 point grid & 32759 & 21843 & 11380 & 3964 \\
  200 point grid &  32761 & 21845 & 11382 & 3967  \\
 \hline
  Theoretical price & 32760 & 21844 & 11381& 3966 \\
  \hline
\end{tabular}
\caption{\em  \label{Tab:LS1}Comparison for a call strip (no global
constraints)}
\end{table}
\end{center}

Table \ref{Tab:LS2} presents the results obtained with the global
constraints $Q_{min}=1300$ and $Q_{max}=1900$. Volume constraints
are presented on Figure~\ref{Fig:Volume}.

\begin{center}
\begin{figure}[h!]
\begin{center}
  \includegraphics[width=8.5cm]{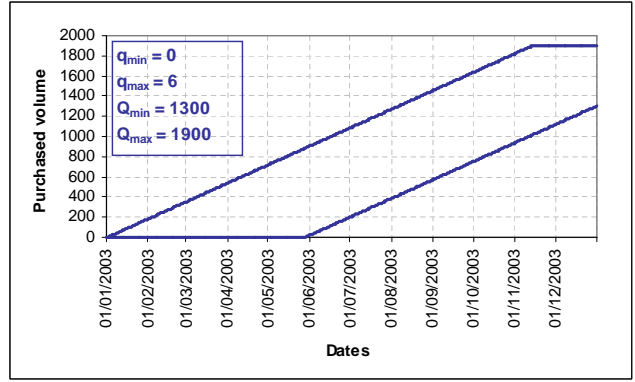}\\
   \caption{\em \label{Fig:Volume}Consumption constraints}
  \end{center}
\end{figure}
\end{center}

\begin{center}
\begin{table}[h!]
\begin{center}
\begin{tabular}{|l|c|c|c|c|}
    \hline
   & $K=5$ & $K=10$ & $K=15$ & $K=20$ \\
   \hline
 Longstaff-Schwartz & [29068,;29758] & [19318,;19993] & [10265,;10892] & [2482,;3038]\\
 \hline
  10 point grid & 29696 & 20216 & 10981& 3067 \\
  20 point grid & 29494& 20018 & 10841 & 2863 \\
  50 point grid & 29372 & 19895 & 10729 & 2718 \\
  100 point grid & 29348 & 19872 & 10704 & 2687 \\
  200 point grid & 29342 & 19866 & 10698 & 2680  \\
  \hline
\end{tabular}
\end{center}
\caption{\em \label{Tab:LS2}Comparison with constraints}
\end{table}
\end{center}

The results seem consistent for both methods, the price given by
quantization always belongs to the confidence interval of the
Longstaff-Schwartz method. One can note that it is true even for
small grids, which means that quantization gives quickly a good
price approximation. Moreover, the price given by quantization is
very close of the theoretical price in the case of a call strip.

\FloatBarrier

\subsubsection{Execution time}
\label{Section:Time} In this section are compared the execution
times to price swing options using optimal quantization and
Longstaff-Schwartz method.

The size of the quantization grid is $100$ for the pricing part and
$200$ for the transitions computation. And $1000$ Monte Carlo
simulations are used. The maturity of the contract is one year.\\

The computer that has been used has the following characteristics:

\medskip
\centerline{Processor: Celeron; CPU 2,4 Ghz; 1,5 Go of RAM; Microsoft Windows 2000.}

\medskip
The execution times given in Table~\ref{Tab:1contrat} concern
the pricing of one contract, which yields the building of the
quantization tree and the pricing using dynamic programming for
quantization.
\\
\begin{center}
\begin{table}[h!]
\centering
\begin{tabular}{|c|c|c|}
  \hline
  Longstaff-Schwartz & Quantization:& Quantization:\\
  & Quantization tree building + Pricing & Pricing only\\
  \hline
  160 s & 65 s & 5 s\\
  \hline
\end{tabular}
\caption{\em \label{Tab:1contrat}Execution time for the pricing of
one contract}
\end{table}
\end{center}

\FloatBarrier
If we consider the pricing of several contracts, there is no need
for re computing the quantization tree if the underlying price model
has not changed. That is why quantization is really faster than
Longstaff-Schwartz in this case, as one can note from the results
presented in Table~\ref{Tab:10contrats}.

\begin{center}
\begin{table}[h!]
\centering
\begin{tabular}{|c|c|}
  \hline
  Longstaff-Schwartz & Quantization \\
  \hline
  1600 s & 110 s \\
  \hline
\end{tabular}
 \caption{\em \label{Tab:10contrats}Execution time for the pricing of 10 contracts}
\end{table}
\end{center}

\subsubsection{Sensitivity Analysis}
When contracts such as swing options ought to be signed,
negotiations usually concern the volume constraints. That is why the
valuation technique has to be very sensitive and coherent to
constraints variation. In this section we will compare the
sensibility to global constraints for Longstaff-Schwartz method and
optimal quantization.

Figure~\ref{Fig:Sensibility2constraints} represents the price of the
contract with regards to the global constraints $Q_{min}$ and
$Q_{max}$, and Figure~\ref{Fig:Coupe} represents the price versus
$Q_{max}$ for a fixed value of $Q_{min}$ equal to $1300$.

\begin{center}
\begin{figure}[h!]
\begin{center}
  \includegraphics[width=10cm]{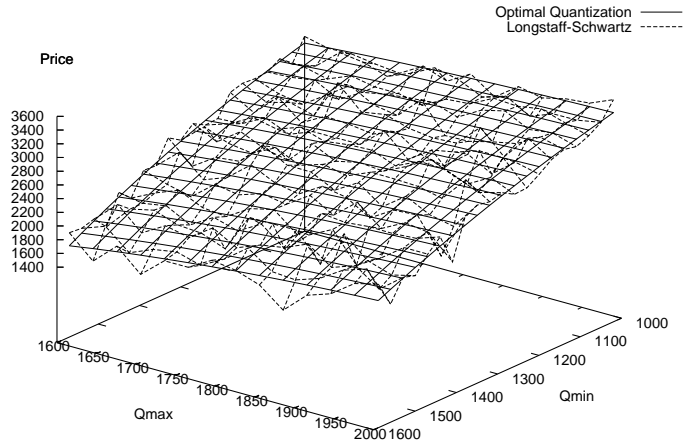}\\
   \caption{\em \label{Fig:Sensibility2constraints}Sensitivity to global constraints}
  \end{center}
\end{figure}
\end{center}

\begin{center}
\begin{figure}[h!]
\begin{center}
  \includegraphics[width=10cm]{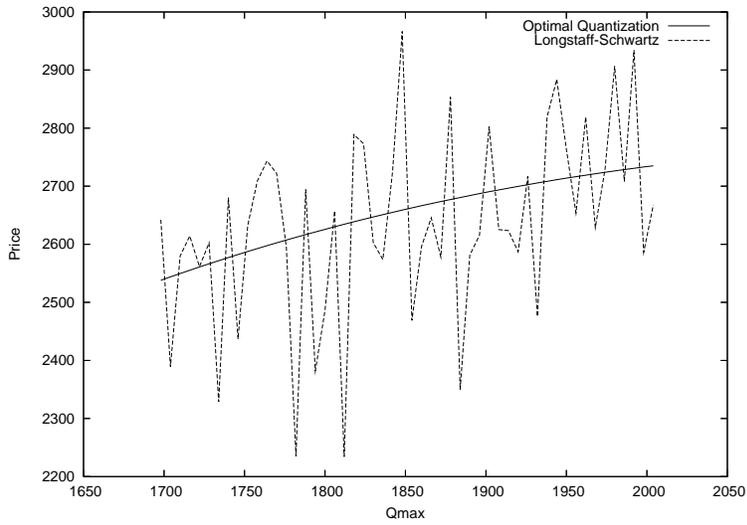}\\
   \caption{\em \label{Fig:Coupe}Sensibility to $Q_{max}$}
  \end{center}
\end{figure}
\end{center}

One can notice that the surface obtained with optimal quantization
is very smooth. If $Q_{max}$ increases, the price increases.
However, it is not always true with Longstaff-Schwartz because of
the randomness of the method, and the limited number of Monte Carlo
simulations imposed by the dimension of the problem and the number
of time steps.

New Monte Carlo simulations are done for each different contract,
$i.e$ each time $Q_{min}$ or $Q_{max}$ varies. Of course, the same
simulations could be used to price all the contracts, but
unfortunately these simulations could be concentrated in the
distribution queues and give a price far from the real one for all
the contracts. As concerns quantization, the grid is build in order
to give a good representation of the considered random variable. One
of the great advantages of optimal quantization over Monte-Carlo is
that this first algorithm always approximates the whole distribution
of the payoff meanwhile it can take a while before Monte-Carlo
explores some parts of it.

\FloatBarrier

\subsubsection{Convergence}
\label{Section:Cvg1F} In this section we will study the convergence
of the quantization method. We focus on the convergence of the
pricing
part of the algorithm.\\

We consider a one year maturity contract with the volume constraints
depicted on Figure~\ref{Fig:Volume}, and the daily forward curve
depicted on Figure~\ref{Fig:Fwd}.

\begin{center}
\begin{figure}[h!]
\begin{center}
  \includegraphics[width=8.5cm]{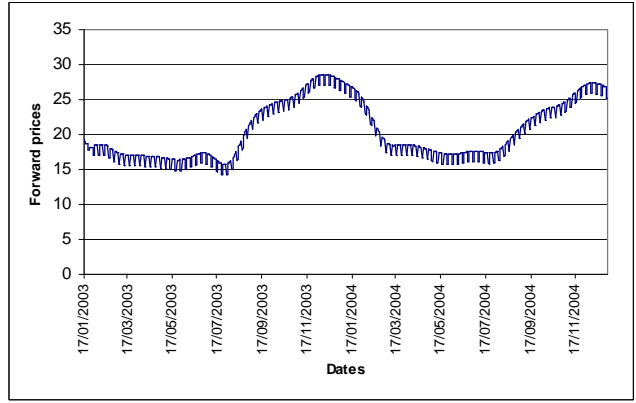}\\
   \caption{\em \label{Fig:Fwd}Daily forward curve}
  \end{center}
\end{figure}
\end{center}

Let $P(N)$ be the price obtained for a quantization grid of size
$N$, the error has been computed as $|P(N)-P(400)|,N \leq 400$. We
assume that the error can be written as a functional of the grid
size $N$ with the following shape:
 $$
N \mapsto \frac{C}{N^{\alpha}}.
$$
A linear regression in a logarithmic scale is done to find the
functional that best fits the empirical error. The $\alpha$
coefficient obtained is $1.96$.\\

Figures \ref{Fig:Cvg1F} and \ref{Fig:Logscale} show the obtained
numerical convergence and the corresponding fitted functional.
\begin{center}
\begin{figure}[h!]
\begin{center}
  \includegraphics[width=10cm]{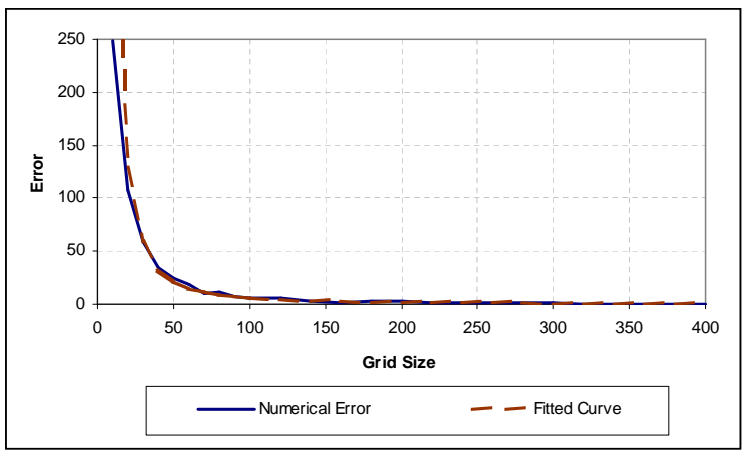}\\
   \caption{\em \label{Fig:Cvg1F}Numerical Convergence}
  \end{center}
\end{figure}
\end{center}

\begin{center}
\begin{figure}[h!]
\begin{center}
  \includegraphics[width=10cm]{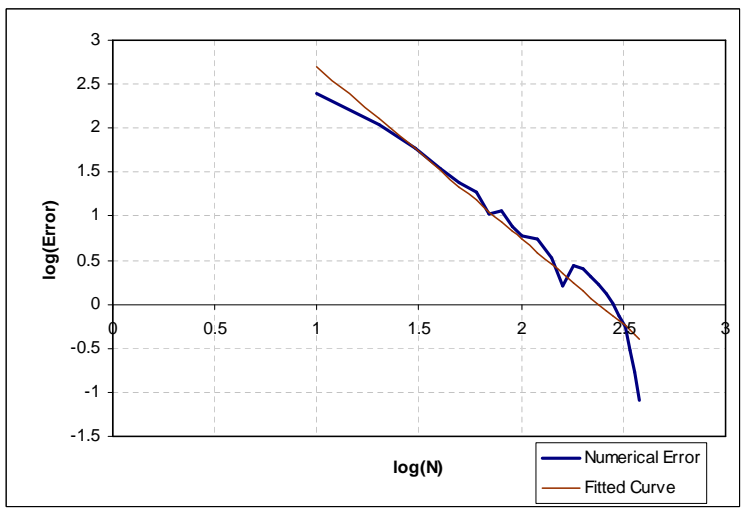}\\
   \caption{\em \label{Fig:Logscale}Numerical Convergence (logarithmic scale)}
  \end{center}
\end{figure}
\end{center}

\FloatBarrier
The same experiments have been done for other
contracts, results are presented in Table \ref{Tab:Cvg1F}. $q_{min}$
and $q_{max}$ are set to $0$ and $6$.

\begin{center}
\begin{table}[h!]
\begin{tabular}{|c|c|c|c|}
  \hline
  Forward Curve & Strike & Constraints ($Q_{min}-Q_{max}$)& Estimated $\alpha$ \\
  \hline
  Figure~\ref{Fig:Fwd} & 20 & 1300-1900 & 1.96 \\
  Flat (20) & 20 & 1300-1900 & 2.07 \\
   Flat (20) & 10 & 1300-1900 & 2.32 \\
  Flat (20) & 20 & 1000-2000 &  1.95\\
  Flat (20)  & 20 & 1600-1800 &  2.26\\
  \hline
\end{tabular}
\caption{\em \label{Tab:Cvg1F}Estimation of the convergence rate}
\end{table}
\end{center}

\FloatBarrier We can conclude that the convergence rate of the
quantization algorithm for pricing swing options is close to
$O(\frac{1}{N^{2}})$. This convergence rate is much better than
Monte-Carlo, and leads to think that optimal quantization is an
efficient alternative to Longstaff-Schwartz method for this problem.

\subsection{Two factor model}
\label{Section:2d} We consider the following diffusion model for the
forward contracts $(F_{t,T})_{0 \leq t \leq T}$:
$$
\frac{dF_{t,T}}{F_{t,T}} = \sigma_1 e^{-\alpha_1 (T-t)}dW^1_t +
\sigma_2 e^{-\alpha_2 (T-t)}dW^2_t
$$
where $W^1$ and $W^2$ are two Brownian motions with correlation
coefficient $\rho$.

Standard computations based on It\^o formula yield
$$
S_t = F_{0,t}\exp{\left(\sigma_1 \int_0^t{e^{-\alpha_1(t-s)}dW^1_s}
+ \sigma_2 \int_0^t{e^{-\alpha_2(t-s)}dW^2_s}-
\frac{1}{2}\Lambda_t^2\right)}
$$
where
$$\Lambda_t^2 = \frac{\sigma_1^2}{2\alpha_1}(1-e^{-2\alpha_1 t}) + \frac{\sigma_2^2}{2\alpha_2}(1-e^{-2\alpha_2 t})
+2\rho\frac{\sigma_1
\sigma_2}{\alpha_1+\alpha_2}(1-e^{-(\alpha_1+\alpha_2) t}).
$$

Unlike the one factor model, the spot price process obtained from
the two factor model is not a Markov process. Hence the dynamic
programming equation (\ref{Eq:PgmDyn}) cannot be used directly.
However, the structure process of the two factor model (See Annex,
page \pageref{Section:2factor})
$$X_t = \left(\int_0^t{e^{-\alpha_1(t-s)}dW^1_s},\int_0^t{e^{-\alpha_2(t-s)}dW^2_s}  \right)$$
is a Markov process, and $S_t=f(X_t)$ where $f:\mathbb{R}^2 \mapsto
\mathbb{R}$ is a continuous function. So we can rewrite the dynamic
programming equation as follows:
\begin{equation}
\left\{
  \begin{array}{l}
    P(t_k,X_{t_k},Q_{t_k})=\max_{q \in [q_{min},q_{max}]}{\{q(f(X_{t_k})-K) +
\mathbb{E}(P(t_{k+1},X_{t_{k+1}},Q_{t_k}+q)|X_{t_k})\}}, \\
\ \\
    P(T,X_T,Q_T)= P_{_T}(f(X_T),Q_T).
  \end{array}
\right.
\end{equation}

Then we need to quantize the $\mathbb{R}^2$ valued structure process
$(X_k)_{k \geq 0}$. This can be done using the Fast Parallel
Quantization (See Annex). The transitions are computed using Monte
Carlo simulations and importance sampling
(see~(\ref{Eq:ImportanceSampling})).

\subsubsection{Call strip}
We first consider a case without constraints, and compare the
results with the theoretical price of the call strip, for several
values of the strike $K$. The maturity of the contract is one month.
Transitions have been computed with $3 000 000$ of Monte-Carlo
simulation. The parameters of the two factor model are:
\begin{equation}\label{Tab:Param}
  \sigma_1=36\%, \alpha_1=0.21, \sigma_2=111\%, \alpha_2=5.4, \rho=-0.11.
\end{equation}
 Table~\ref{Tab:Callstrip} presents the results obtained for a strip
of call. Even if the quantized process is taking values in
$\mathbb{R}^2$, prices are close to the theoretical price even for
small grids.
\begin{center}
\begin{table}[h!]
\begin{center}
\begin{tabular}{|c|c|c|c|c|}
    \hline
   & $K=5$ & $K=10$ & $K=15$ & $K=20$ \\
   \hline
 Theoretical Price & 2700 & 1800.21 & 924.46 & 268.59\\
 \hline
  50 point grid & 2695.26 & 1795.26 & 918.18& 261.06 \\
  100 point grid & 2699.79& 1799.89 & 923.64 & 267.17 \\
  200 point grid & 2697.67 & 1797.83 & 921.94 & 266.71 \\
  300 point grid & 2702.02 & 1800.16 & 924.40 & 268.52 \\
  \hline
\end{tabular}
\end{center}
\caption{\em \label{Tab:Callstrip}Call Strip}
\end{table}
\end{center}

\subsubsection{Convergence}
\label{Section:Cvg2F}
%
We use the same procedure as in section \ref{Section:Cvg1F} to find
the functional $N \mapsto \frac{C}{N^{\alpha}}$ that best fits the
empirical error.

Figure~\ref{Fig:CvgMarkov} shows an example of the empirical error
and table \ref{Tab:Cvg2F} gather the values of $\alpha$ obtained for
different contracts. The contract maturity has been set to one
month.

\begin{center}
\begin{figure}[h!]
\begin{center}
  \includegraphics[width=10cm]{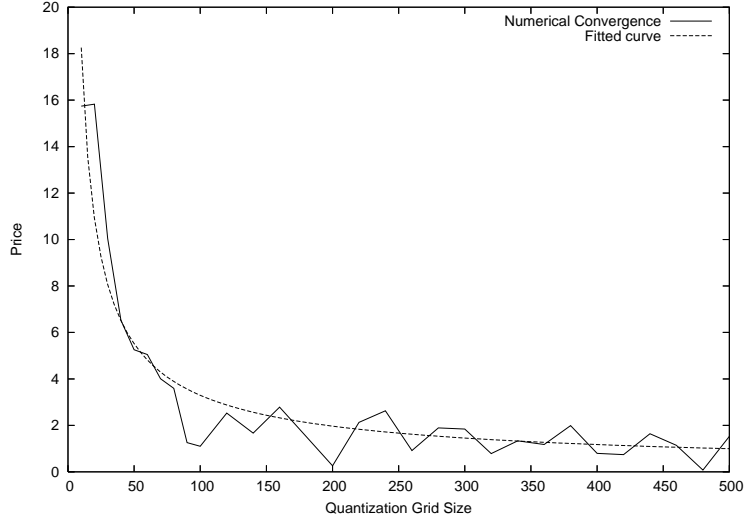}\\
   \caption{\em \label{Fig:CvgMarkov}Numerical Convergence}
  \end{center}
\end{figure}
\end{center}

\begin{center}
\begin{table}[h!]
\begin{tabular}{|c|c|c|c|}
  \hline
  Forward Curve & Strike & Constraints ($Q_{min}-Q_{max}$)& Estimated $\alpha$ \\
  \hline
  Flat (20) & 10 & 80-140 & 1.26\\
  Flat (20) & 20 & 80-140 & 1.00\\
  Flat (20)  & 20 & 30-170 & 0.67\\
  Flat (20) & 20 & 100-120 & 1.19\\
  \hline
\end{tabular}
\caption{\em \label{Tab:Cvg2F}Estimation of the convergence rate}
\end{table}
\end{center}

The convergence of the quantization algorithm is close to
$O(\frac{1}{N})$. The convergence rate is linked to the dimension
$d$ of the structure process, and from the results obtained in
section \ref{Section:Cvg1F} and in this section, we can assume that
the convergence rate is close to $O(\frac{1}{N^{2/d}})$, which is
better than the error bound theoretically established in Section
\ref{Section:Cvg}.

\FloatBarrier

\subsection{Dimension reduction}\label{dimred}
In the case of multi-factorial models, we need to quantize the
structure process $(X_k)_k$ instead of the spot process
$S_{t_k}=f(X_k)$ in order to work with a Markov process (See section
\ref{Section:2d}). Only the one factor model is Markovian. That is
why quantization and Longstaff-Schwartz method have been compared
just for this model. Longstaff Schwartz method also requires a
Markov underlying process.

From an operational point of view, it is interesting to study the
results obtained by formally quantizing the spot process
$(S_{t_k})$, regardless to its dynamics, and using the approximation
$$
\mathbb{E}(X|\mathcal{F}_{t_k})\simeq \mathbb{E}(X|S_{t_k})
$$
for any random variable $X$, even if the spot price $(S_{t_k})$ is
not a Markov process. Similar approximation has already been
proposed by Barraquand-Martineau in~\cite{BAMA}. Numerical tests
have shown that the resulting prices remain very close to those
obtained by quantizing the structure process in the case of a two
factor model. Execution time and convergence rate are significantly
faster, and the quantization tree can be computed as presented in
the Annex page~\pageref{Section:AR1} using
Equation~(\ref{Eq:ImportanceSampling1D}).

Even if there is no theoretical evidence on the error, this approach
seems useful to get quick results. Table \ref{Tab:Comparison}
presents some results, parameters of the two factor model are those
of~(\ref{Tab:Param}) and volume constraints are represented on
Figure~\ref{Fig:Volume}.

\begin{center}
\begin{table}[h!]
\begin{center}
\begin{tabular}{|l|c|c|c|c|}
    \hline
   Quantized Process & $K=5$ & $K=10$ & $K=15$ & $K=20$ \\
   \hline
 Spot Price & 30823.38 & 21414.30 & 13021.87 & 5563.68\\
 \hline
Structure Process &  30705.78 & 21518.75& 13123.56 & 5722.89 \\
  \hline
\end{tabular}
\end{center}
\caption{\em \label{Tab:Comparison}Quantization of the spot price vs
Quantization of the bivariate structure process}
\end{table}
\end{center}

The convergence rate obtained in this case is always
$O(\frac{1}{N})$, which is consistent with the general rate
$O(\frac{1}{N^{2/d}})$, because the quantized process is
$\mathbb{R}$-valued.

Therefore, even if the spot process is not Markov, the quantization
method can be performed all the same way as if it were, with small damage in practice. This is of
course not a general consideration but rather an observation over
the considered problem. The dramatic increase in the computation
effort that can be gained from this observation can justify in this
case a lack of formal rigor.

\section*{Conclusion}

In this article, we have introduced an optimal quantization method
for the valuation of swing options. These options are of great
interest in numerous modeling issues of the energy markets and their
accurate pricing is a real challenge.\\
Our method has been compared to the famous Longstaff-Schwartz
algorithm and seems to perform much better on various examples. In
fact, the optimal quantization method shares the good properties of
the so called tree method but is not limited by the dimension of the
underlying. Moreover, specific theoretical results provide {\it a
priori} estimates on the error bound of the method. \\
Thus, optimal quantization methods suit very well to the valuation
of complex derivatives and further studies should be done in order
to extend the present results to other structured products arising
in the energy sector.


\bigskip
\centerline{\bf \large Annex: Fast Parallel quantization (FPQ)}
\label{Section:FPQ}

\medskip In this annex we propose an efficient method to
quantize a wide family of spot price dynamics $(S_{t_k})$. To be
precise we will assume that a time discretization $\Delta$ being
fixed,
\begin{equation}
S_{k \Delta} = f_k(X_k),k \geq 0 \label{Eq:Spot}
\end{equation}
where $(X_k)_{k \geq 0}$ is a $\mathbb{R}^m$-valued Gaussian
auto-regressive process and $(f_k)_{0 \leq k \leq n}$ a family of
continuous functions. The fast quantization method applies to the
Gaussian process $(X_k)_{0 \leq k \leq n}$. We will apply it to a
scalar two factor model in full details. As a conclusion to this
section we will sketch the approach to a multi-factor model.

\subsubsection*{Quantization of the Gaussian structure process}
\label{Section:AR1} We consider a centered Gaussian first order
auto-regressive process in $\mathbb{R}^m$: \begin{equation}X_{k+1}
= AX_k + T \varepsilon_{k+1} \label{Eq:Y}\end{equation} where $A \in
\mathcal{M}(m\times m,\mathbb{R})$, $T \in \mathcal{M}(m\times
m,\mathbb{R})$ lower triangular, and $(\varepsilon_k)$
i.i.d. with $\mathcal{N}(0,I_m)$ distribution.\\

Denote by $D(Z) = \left[ \mathbb{E}(Z_iZ_j) \right]_{1\leq i,j\leq
m}$ the covariance matrix of $Z$. We have $\forall k \in \mathbb{N}$:
$$
D(X_{k+1})  =  AD(X_k)A^* + TT^*.
$$
Denote $\Sigma_k$ the lower triangular matrix such that
$D(X_k)=\Sigma_k \Sigma_k^*$.

\medskip We consider for every $k=0,\ldots,n-1$, an {\em optimal (quadratic) quantizer} $x_k^{(N_k)}$  of size $N_k$,
for the $\mathcal{N}(0,I_m)$ distribution. The quantization grid of
the random variable $X_k$ is taken as a dilatation of
$x_k^{(N_k)}$,$i.e$
$$
\bar{x}_{k} = \Sigma_k \,x_k^{(N_k)} := (\Sigma_k x^{(N_k),i})_{1 \leq i \leq N_k}.
$$

To calculate the conditional expectations in the dynamic programming
equation, we need to get the following transition probabilities:
$$
\pi^{ij}_k=\mathbb{P}(X_{k+1} \in C_j(\bar{x}_{k+1}) | X_k \in
C_i(\bar{x}_{k}))
$$
where $C_i(x)$ denotes the $i$-th Voronoi cell
of the generic quantizer $x\!\in (\mathbb{R}^d)^N$. Then
$$
\mathbb{P}(X_k \in C_i(\bar{x}_{k})) = \mathbb{P}(Z \in
C_i(x_k^{(N_k)})),
$$
with $Z \sim \mathcal{N}(0,I_m)$. This probability
is provided as a companion parameter with the normal distribution grid files (available on~\cite{website06}).

To get the transition probability $\pi_k^{ij}$ we need to compute
$$
p^{ij}_k=\mathbb{P}(X_{k+1} \in C_j(\bar{x}_{k+1}), X_k \in
C_i(\bar{x}_{k})).
$$
\begin{prop} Let $X$ be a discrete time process described as
above. Let $U$, $V$ be two gaussian random variables
$\mathcal{N}(0,I_m)$. Then we have for every $k \in \{0,\dots,n-1\}$,
every $i \in \{1,\dots,N_k\}$, every $ j \in \{1,\dots,N_{k+1}\}$,
\begin{equation}\mathbb{P}(X_{k+1} \in C_j(\bar{x}_{k+1}) , X_k \in
C_i(\bar{x}_{k})) = \mathbb{P}(U \in C_i(x_k^{(N_k)}),A_{k+1}U +
B_{k+1}V \in C_j(x_{k+1}^{N_{k+1}}))\label{Eq:Transition}
\end{equation}
where $A_k$ and $B_k$ are $q\times q$ matrices whose coefficients depend
on $k$, and on the matrices $A$, $T$.
If $k=0$,
\begin{equation}
\mathbb{P}(X_{1} \in C_j(\bar{x}_{1})) = \mathbb{P}(V \in
C_j(x_1^{(N_1)})).
\end{equation}
\end{prop}

\begin{proof} We have:
$$
\mathbb{P}(X_{k+1} \in C_j(\bar{x}_{k+1}); X_k \in C_i(\bar{x}_k))=
\mathbb{P}(AX_k + T\varepsilon_{k+1}\in C_j(\bar{x}_{k+1}) ; X_k \in
C_i(\bar{x}_k)).
$$
We consider the couple $(X_k, T\varepsilon_{k+1})$.
$T\varepsilon_{k+1}$ is independent of $X_k$. Let $\eta = (\eta_1,
\eta_2)$ a couple of independent Gaussian random vectors: $\eta_i \sim N(0,I_m),
i=1,2$. Then $\left( X_k,T\varepsilon_{k+1}\right)\sim\left(
\Sigma_k\eta_1,T\eta_2\right) $ and
\begin{eqnarray*}
\mathbb{P}(X_{k+1} \in C_j(\bar{x}_{k+1}); X_k \in C_i(\bar{x}_k)) &
= & \mathbb{P}(A\Sigma_k\eta_1 + T\eta_2\in C_j(\bar{x}_{k+1}) ;
\Sigma_k\eta_1 \in C_i(\bar{x}_k)) \\
& = & \mathbb{P}(\Sigma_{k+1}^{-1}\left(A\Sigma_k\eta_1 + T\eta_2
\right) \in C_j(x_{k+1}^{(N_{k+1})}) ;\eta_1 \in C_i(x_k^{(N_k)})).
\end{eqnarray*}
Setting
$$
A_{k+1} = \Sigma_{k+1}^{-1}A\Sigma_k, \quad  B_{k+1}= \Sigma_{k+1}^{-1}T
$$
we get
$$
\mathbb{P}(X_{k+1} \in C_j(\bar{x}_{k+1}); X_k \in C_i(\bar{x}_k))
=\mathbb{P}(A_{k+1}\eta_1 + B_{k+1}\eta_2 \in C_j(x_{k+1}^{(N_{k+1})}) ;
\eta_1 \in C_i(x_k^{(N_k)})).
$$

If $k=0$ and $\Sigma_0 \equiv 0$, the quantity
$$\mathbb{P}(X_{k} \in C_j(\bar{x}_{k})) = \mathbb{P}(\eta_1 \in
C_j(x_k^{(N_{k})}))$$ is given as a companion parameter with the
quantization grids of the normal distribution.
\end{proof}

\begin{rem}
Equation (\ref{Eq:Transition}) emphasizes the fact that the
transitions can be computed in parallel.
\end{rem}

\begin{rem}
To simplify the structure of the quantization tree we propose to
consider the same normalized grid of size $N_k = N$ at each step $k$
but other choices are possible like those recommended in
\cite{bernouilli}.
\end{rem}

\subsubsection*{Numerical methods}

Hereafter we will focus on the numerical computation of these transitions.

\medskip
\noindent $\bullet$ {\sc The standard Monte Carlo approach} The simplest way is to use a Monte Carlo method. One just
  needs to simulate couples of independent gaussian random
  variables $(\eta_1,\eta_2)$. This approach can be used
  whatever the dimension $m$ of the random variables $\eta_1$ and
  $\eta_2$ is.  It can clearly be parallelized as any MC simulation but fail to estimate the transition form states which are not often visited.

\medskip
\noindent $\bullet$ {\sc Fast  Parallel Quantization Method}  In order to improve the accuracy, especially for the points
  $x^{(N_{k+1}),j}$ of the grids $x_{k+1}^{(N_{k+1})}$ which are rarely
  reached by the paths starting from the cell of $x^{(N_{k}),i}$,
  it is possible to perform importance
  sampling.  The idea is to use Cameron-Martin formula to re-center the
  simulation: for every $k\!\in\{0,\ldots,n-1\}$ and every $i\!\in\{1,\ldots,N_k\}$,
\begin{multline}\label{Eq:ImportanceSampling}
p^{ij}_k=\mathbb{P}\left(\eta_1 \in C_i(x_k^{(N_k)});A_{k+1}\eta_1 +
B_{k+1}\tilde{\eta}_2  \in C_j(x_{k+1}^{(N_{k+1})})\right) = \\
e^{-\frac{1}{2}|x_k^{(N_k),i}|^2}\mathbb{E}\left(e^{-(x_k^{(N_k),i})^*
\eta_1}\mbox{\bf 1}_{\{\eta_1+x_k^{(N_k),i} \in
C_i(x_k^{(N_k)})\}}\mbox{\bf 1}_{\{ A_{k+1}(\eta_1+x_k^{(N_k),i})
+B_{k+1}\eta_2 \in C_j(x_{k+1}^{(N_{k+1})})\}} \right).
\end{multline}
Then  these expectations can be computed by Monte Carlo simulations, the transitions between the different times steps can be computed in parallel.

\medskip
\noindent $\bullet$ {\sc Quantized Parallel Quantization Method} If $m =1$, the transitions can be computed using again optimal
quantization, because in low dimension (say $d \leq 4$),
quantization converges faster than Monte Carlo method. In this case,
we have to compute a two dimensional expectation.

We estimate for every $k\!\in
\{0,\dots,n-1\}$,  every $i \!\in \{1,\dots,N_k\}$ and every $ j\, \in
\{1,\dots,N_{k+1}\}$  the following probabilities:
\begin{equation}p_k^{ij}=\mathbb{P}(\eta_1 \in
C_i(x_k^{(N_k)});\alpha_{k+1}\eta_1 + \beta_{k+1}\eta_2  \in
C_j(x_{k+1}^{(N_{k+1})}))\label{Eq:Proba1d}\end{equation} where
$(\eta_1,\eta_2) \sim \mathcal{N}(0,I_2)$, and $\alpha_k$ and
$\beta_k$ are scalar coefficients satisfying
$\alpha_k^2+\beta_k^2=1$:
\begin{eqnarray*}
\alpha_k & = \frac{\Sigma_kA}{\Sigma_{k+1}}& \\
\beta_k &= \frac{T}{\Sigma_{k+1}}.&
\end{eqnarray*}

To alleviate notations, we temporarily set $x=x_k^{(N_k)}$ and
$y=x_{k+1}^{(N_{k+1})}$.

We define $ \left[x^{i-\frac{1}{2}},x^{i+\frac{1}{2}} \right]=
\left[\frac{1}{2}(x^i + x^{i-1}), \frac{1}{2}(x^i + x^{i+1})\right]
= C_i(x)$ (the same shortcut is implicitly defined  for $y$).

In order to reduce the problem dimension, it is possible to write
the probability $p_{ij}^k$ as a double integral, and to integrate
first with respect to the second variable by using Fubini theorem:
\begin{multline*}\mathbb{P}(\eta_1 \in C_i(x);\;\alpha_{k+1}\eta_1 + \beta_{k+1}\eta_2
\in C_j(y)) = \\
\mathbb{E}\left(\mbox{\bf 1}_{\{x^{i-\frac{1}{2}}\leq \eta_1 \leq
x^{i+\frac{1}{2}}\}} \left(\mathcal{N}(\frac{y^{j-\frac{1}{2}} -
\alpha_{k+1}\eta_1}{\beta_{k+1}}) -
\mathcal{N}(\frac{y^{j+\frac{1}{2}} -
\alpha_{k+1}\eta_1}{\beta_{k+1}})\right)\right)
\end{multline*}
where $\mathcal{N}(x)$
is the distribution function of the normal distribution.

Importance sampling can again be used to improve the results
precision. Eventually we have to compute the following one
dimensional expectation:
\begin{multline}
\mathbb{P}\left(\eta_1 \in C_i(x);\alpha_{k+1}\eta_1 +
\beta_{k+1}\tilde{\eta}_2  \in C_j(y)\right) =  \\
e^{\frac{-(x^i)^2}{2}}\mathbb{E}\left[e^{-x^i
\eta_1}\mbox{\bf 1}_{\left\{ -\frac{\Delta x^i}{2}\leq \eta_1 \leq
\frac{\Delta x^{i+1}}{2}\right\}}
\left\{\mathcal{N}\left(\frac{y^{j-\frac{1}{2}} -
\alpha_{k+1}(\eta_1
+x^i)}{\beta_{k+1}}\right)\right.\right. \\
\left.\left.- \mathcal{N}\left(\frac{y^{j+\frac{1}{2}} -
\alpha_{k+1}(\eta_1+x^i)}{\beta_{k+1}}\right)\right\}\right].
\label{Eq:ImportanceSampling1D}
\end{multline}
For this one-dimensional expectation computation, quantization can be used again
since  it converges faster than Monte Carlo method.

\subsubsection*{Example: Two factor model}
\label{Section:2factor} We consider the following diffusion model
for the forward contracts $(F_{t,T})_{0 \leq t \leq T}$:
$$
\frac{dF_{t,T}}{F_{t,T}} = \sigma_1 e^{-\alpha_1 (T-t)}dW^1_t + \sigma_2 e^{-\alpha_2 (T-t)}dW^2_t
$$
where $W^1$ and $W^2$ are two Brownian motions with correlation
coefficient $\rho$.

Standard computations based on It\^o formula yield
$$
S_t =
F_{0,t}\exp{\left(\sigma_1 \int_0^t{e^{-\alpha_1(t-s)}dW^1_s} +
\sigma_2 \int_0^t{e^{-\alpha_2(t-s)}dW^2_s}-
\frac{1}{2}\Lambda_t^2\right)}
$$
where
$$\Lambda_t^2 = \frac{\sigma_1^2}{2\alpha_1}(1-e^{-2\alpha_1 t}) + \frac{\sigma_2^2}{2\alpha_2}(1-e^{-2\alpha_2 t})
+2\rho\frac{\sigma_1
\sigma_2}{\alpha_1+\alpha_2}(1-e^{-(\alpha_1+\alpha_2) t}).
$$
We have $S_t =F_{0,t}\exp{\left(\sigma_1 X^1_t + \sigma_2 X^2_t-
\frac{1}{2}\Lambda_t^2\right)}$, where $X_t$ is the following
structure process:
\begin{equation}
X_t =
\left(\int_0^t{e^{-\alpha_1(t-s)}dW^1_s},\int_0^t{e^{-\alpha_2(t-s)}dW^2_s}
\right). \label{Eq:Sructure2F}
\end{equation}

\begin{prop}
\label{Prop:OU} Let $Z=(Z_t)$ be an Ornstein-Uhlenbeck process. $Z_t =
Z_0 +  \int_0^t{e^{-\alpha(t-s)}dB_s}$ where $B$ is a standard
Brownian motion, and $Z_0$ is Gaussian and independent of $B$. $Z$
can be written at discrete times $k\Delta$ as a first order
auto-regressive process:
\begin{equation} \label{Eq:OU_AR1}Z_{k+1} = e^{-\alpha\Delta}Y_k +
\sqrt{1-e^{-2\alpha
\Delta}}\sqrt{\frac{1}{2\alpha}}\varepsilon_{k+1} \end{equation}
where $(\varepsilon_k)$ is i.i.d and $\varepsilon_1 \sim
\mathcal{N}(0,1)$.
\end{prop}

$X_t$ is made up with two Ornstein-Uhlenbeck processes. Using
Proposition \ref{Prop:OU}, it yields:

\begin{prop}
$$X_{k+1} = AX_k + T \varepsilon_{k+1}$$
with $\varepsilon_k \sim \mathcal{N}(0,I_2)$ i.i.d. and

$$A = \left[
\begin{array}{cc}
e^{-\alpha_1 \Delta} & 0\\
 0 & e^{-\alpha_2 \Delta} \\
\end{array}
\right]$$

$$r = \rho \frac{\frac{1}{\alpha_1 + \alpha_2}(1-e^{-(\alpha_1 + \alpha_2)\Delta})}{\sqrt{\frac{1}{4\alpha_1\alpha_2}(1- e^{-2\alpha_1 \Delta})(1- e^{-2\alpha_2
\Delta})}}$$

$$T = \left[
\begin{array}{cc}
 \frac{1}{2\alpha_1}(1- e^{-2\alpha_1 \Delta})& 0 \\
  \frac{1}{2\alpha_2}(1- e^{-2\alpha_2 \Delta})r& \frac{1}{2\alpha_2}(1- e^{-2\alpha_2 \Delta})\sqrt{1-r^2}
\end{array}
\right].$$

\end{prop}

Hence it is possible to use the fast parallel quantization method
described in section \ref{Section:AR1}.

\subsubsection*{General multi-factor Gaussian model}
\label{Section:GeneralCase}

More generally, we consider a family of price dynamics that can be
written as follows:
\begin{equation}
\label{Eq:Model} \frac{dF_{t,T}}{F_{t,T}} = \sum_{i=1}^m P_i(T-t)
e^{-\alpha_i (T-t)}dW^i_t \end{equation} where $P_i(x)$ is a
polynomial function of degree $d_i$, for every $i=1,\dots,m$, and $W$
is a Brownian motion, with $d<W^i, W^j>_t=\rho_{ij}dt$.

The two factor model (Section \ref{Section:2factor}) corresponds to
$m=2$, $P_i\equiv \sigma_i$, $i=1,2$.

In order to price a swing option with such a model, we first need to
quantize it. Equation~(\ref{Eq:Model}) yields:
\begin{equation}F_{t,T} = F_{0,T}e^{\sum_{i=1}^m
\int_0^t P_i(T-s) e^{-\alpha_i (T-s)}dW^i_s - \frac{1}{2}
\phi(t,T)}\label{Eq:Structure}
\end{equation}
where
$$
\phi(t,T)=\sum_{i=1}^m \int_0^t P^2_i(T-s) e^{-2\alpha_i (T-s)}ds +
\sum_{i\neq j} \rho_{ij}\int_0^t P_i(T-s)P_j(T-s)
e^{-(\alpha_i+\alpha_j) (T-s)}ds.
$$
Practically we focus on the spot price $F_{t,t}$ or the day-ahead
contract $F_{t,t+1}$. Unfortunately these processes are not
Markovian in a general setting, except when $m=1$ and $d_1=0$ (Ornstein-Uhlenbeck process).

We consider a discretization time step $\Delta > 0$, and we set, for
all $i\in \{1,\dots,m \}$ and for all $l\in \{0,\dots,d_i \}$
$$
X_k^{i,l}=\int_0^{k \Delta}P_i \left( (k+l)\Delta
-s\right)e^{-\alpha_i\left((k+l)-s\right)}dW^i_s.
$$

\begin{prop}
\label{Prop:AR1} $X_k = [X^{i,l}_k ]_{1 \leq i \leq m,0 \leq l \leq
d_i}$ is a $\mathbb{R}^{d_1+\dots+d_m+m}$-valued gaussian AR(1).
\end{prop}

\begin{lem}
Let $P \in \mathbb{R}[Z]$, $d^oP=d$ and $\theta \in \mathbb{R}^*$.
Then $\left(P(Z+l \theta) \right)_{0 \leq l \leq d}$ is a basis of
$\mathbb{R}^d$. \label{Lem:Basis}
\end{lem}

\begin{proof}
Using a dimension argument, only the linear independence of the family has to be
checked. And we have
\begin{eqnarray*}
\sum_{k=0}^d \lambda_k P(Z+k \theta)=0 & \Leftrightarrow &
\sum_{k=0}^d \lambda_k \sum_{j=0}^d \frac{(k
\theta)^j}{j!}P^{(j)}(Z)=0\\
& \Leftrightarrow & \sum_{j=0}^d
\frac{\theta^j}{j!}\left(\sum_{k=0}^d \lambda_k k^j
\right)P^{(j)}(Z)=0.
\end{eqnarray*}
Since  $\left( P^{(j)}(Z) \right)_{0 \leq j \leq d}$ is a basis of
$\mathbb{R}_d[Z]$, it yields
$$
\forall j\!\in\{0,\ldots,d\},\quad\sum_{k=0}^d \lambda_k k^j=0
$$
so that $\lambda_k=0, \;0 \leq k \leq d$ since ${\rm det}[k^j]_{0 \leq k,j \leq d}
\neq 0$ (Vandermonde determinant).
\end{proof}

\begin{proof}(of Proposition \ref{Prop:AR1})
We can extend the definition of $X_k^{i,l}$ to $l \in \mathbb{N}$.
It is easy to check that
\begin{eqnarray*}
X_{k+1}^{i,l} & = & \int_0^{(k+1)\Delta}P_i\left((k+1+l)\Delta-s
\right)e^{-\alpha_i \left((k+1+l)\Delta-s\right)}dW^i_s \\
& = & X_{k}^{i,l+1} + \varepsilon_{k+1}^{i,l}
\end{eqnarray*}
where $\varepsilon_{k+1}^{i,l} =
\int_{k\Delta}^{(k+1)\Delta}P_i\left((k+1+l)\Delta-s
\right)e^{-\alpha_i \left((k+1+l)\Delta-s\right)}dW^i_s$.

If $l=d_i$,
$$
X_{k+1}^{i,d_i}=  X_{k}^{i,d_i+1} + \varepsilon_{k+1}^{i,d_i}.
$$
According to Lemma \ref{Lem:Basis},
$$
P_i\left(Z+(d_i+1)\Delta \right) = \sum_{l=0}^{d_i}
\lambda^{i,l}P_i(Z+l\Delta).
$$
Hence
\begin{eqnarray*}
X_{k}^{i,d_i+1} & = & \int_0^{k\Delta}P_i\left((k+d_i+1)\Delta-s
\right)e^{-\alpha_i \left((k+d_i+1)\Delta-s\right)}dW^i_s \\
& = & \sum_{l=0}^{d_i} \lambda^{i,l}
\int_0^{k\Delta}P_i\left((k+l)\Delta-s \right)e^{-\alpha_i
\left((k+l)\Delta-s\right)}dW^i_s e^{-\alpha_i
(d_i+1-l)\Delta} \\
 & = & \sum_{l=0}^{d_i} \lambda^{i,l}e^{-\alpha_i
(d_i+1-l)\Delta}X_k^{i,l} \\
& = & \sum_{l=0}^{d_i} \tilde{\lambda}^{i,l}X_k^{i,l}.
\end{eqnarray*}

Finally we have
$$
X_{k+1}^{i.} = A^iX^{i.}_k + \varepsilon^{i.}_{k+1}
$$
where
$$
A^{i.}=\left(
\begin{array}{ccccc}
0 & 1 & 0 & \cdots & 0 \\
\vdots & \ddots & \ddots & \ddots & \vdots \\
\vdots& & \ddots & \ddots & 0 \\
0&\cdots & \cdots & 0 & 1\\
\tilde{\lambda}^{i,0}& \cdots & \cdots & \cdots & \tilde{\lambda}^{i,d_i}\\
\end{array}
\right)
$$
and
$$X_{k+1} = AX_k + \varepsilon_k,X_0=0$$
where $\varepsilon_k \in \sigma\left(W^i_u-W^i_{k\Delta},k\Delta
\leq u \leq (k+1)\Delta,i=1,\dots,m\right)$ is independent of $\mathcal{F}^W_{k\Delta}$.
The process $(X_k)_k$ is thus a gaussian AR(1).
\end{proof}

$X_k$ is the structure process for the spot price
$S_{k\Delta}=F_{k\Delta,k\Delta}$. Its dimension is $\sum_{i=1}^m
(d_i+1)$. For the two factor model, the structure process is
$\mathbb{R}^2$-valued, because $m=2$, and $d_i=0,i=1,2$. This is
coherent with (\ref{Eq:Sructure2F}).

\end{document}